 \documentclass[smallabstract,smallcaptions]{dccpaper}

\usepackage{citesort}
\usepackage{amsmath}
\usepackage{amssymb}
\usepackage{amsthm}
\newtheorem{lemma}{Lemma}
\usepackage[dvipsnames,table]{xcolor}
\usepackage{url}
\usepackage{soul}
\usepackage{bold-extra}
\usepackage{tabularx,multirow,multicol,booktabs,array}
\usepackage{enumitem}

\usepackage{silence}
\usepackage[capitalize]{cleveref} 
\crefname{section}{§}{§§}

\usepackage[subrefformat=parens,labelformat=parens]{subfig}

\newlength{\figurewidth}
\newlength{\smallfigurewidth}

\newcolumntype{?}{!{\vrule width 1pt}}

\newcommand{\lcell}{\raggedright\arraybackslash}
\newcommand{\ccell}{\centering\arraybackslash}
\newcommand{\rcell}{\raggedleft\arraybackslash}

\newcommand{\bftab}{\fontseries{b}\selectfont}

\begin{document}

\title
{\large
\textbf{Selective Run-Length Encoding}
}

\author{%
Xutan Peng, Yi Zhang, Dejia Peng, and Jiafa Zhu\\[0.5em]
{\small\begin{minipage}{\linewidth}\begin{center}
Huawei Galois Lab\\
\url{{pengxutan, zhangyi418, pengdejia, zhujiafa}@huawei.com} 
\end{center}\end{minipage}}
\thanks{We would like to express our sincerest gratitude to Liwei Guo, Zhaoyi Sun, Fang Li, Jie Liu, Guanyi Chen, Yi Xin, Ruizhe Li, and Shuwei Qian for their
insightful and helpful comments.}}

\maketitle
\thispagestyle{empty}

\begin{abstract}
Run-Length Encoding (RLE) is one of the most fundamental tools in data compression. However, its compression power drops significantly if there lacks consecutive elements in the sequence. In extreme cases, the output of the encoder may require more space than the input (aka \ul{\textit{size inflation}}). To alleviate this issue, using combinatorics, we quantify RLE’s space savings for a given input distribution. With this insight, we develop the first algorithm that \textit{automatically} identifies suitable symbols, then selectively encodes these symbols with RLE while directly storing the others without RLE. 
Through experiments on real-world datasets of various modalities, we empirically validate that our method, which maintains RLE's efficiency advantage, can effectively mitigate the size inflation dilemma.
\end{abstract}

\section{Introduction}\label{sec:intro}
Run-Length Encoding (RLE) is a simple yet powerful data compression technique that can be used to losslessly encode \textit{runs}, i.e., sequences of repeated symbols, in a compact form~\cite{Golomb1966RunlengthE}. The basic idea behind RLE is to represent each element only once, followed by its count within the run. For instance, \texttt{<\textcolor{ForestGreen}{a}, \textcolor{Blue}{b}, \textcolor{Blue}{b}, \textcolor{Blue}{b}, \textcolor{ForestGreen}{a}, \textcolor{ForestGreen}{a}, \textcolor{red}{c}, \textcolor{red}{c}, ...>} can be stored as \texttt{<\textcolor{ForestGreen}{a}, \textcolor{Blue}{b}, \textcolor{ForestGreen}{a}, \textcolor{red}{c}, ...>} (namely an \textit{encoded variable}) together with \texttt{<\textcolor{ForestGreen}{1}, \textcolor{Blue}{3}, \textcolor{ForestGreen}{2}, \textcolor{red}{2}, ...>} (namely a \textit{run-control variable}). Note that the encoded variable and the run-control variable are equal in length. Binary code for an element in the run-control variable is usually implemented with a fixed width (denoted as $b_r$) and can thus represent a run length $\in [1, 2^{b_r}]$. Should the length (denoted as $n$) exceed $2^{b_r}$, this run will be separated into $\lceil \frac{n}{2^{b_r}}\rceil$ divisions by RLE to avoid overflow~\cite{KLEIN2014326,jang2023zero}.

RLE can significantly reduce the size of the encoded data, especially when dealing with repetitive patterns or long sequences. In addition, it is easy to implement and incurs a low overhead (with the time complexity being $\mathcal{O}(N)$ during both encoding and decoding). As a result, RLE becomes one of the most important compressors in applications where speed and simplicity are key factors: Digital Archive~\cite{rle-preprocess-dcc2021}, Time-Series Database~\cite{sprintz}, and Image Codec (e.g., BMP and JPEG~\cite{images}), to list a few.

Nevertheless, the \textit{vanilla} RLE has its Achilles' heel: if long runs are missing in the input sequence, the output may occupy even more space~\cite{Sayood1996IntroductionTD}. One straightforward solution is to perform two RLE passes. The first pass is \textit{exploratory} and involves all symbols. It identifies symbols \textit{suitable} to RLE, i.e., symbols whose total size in the output (encoded variable and run-control variable) is no larger than the input, in a post-hoc fashion. The second pass, which acts as the actual compression, only encodes symbols discovered in the first pass (i.e. suitable to RLE) and leaves other symbols unaffected. The major downside of this method is the computation cost. In particular, when there are various binary representation options (e.g., \cref{ssec:exp-setup}), finding the best configuration, undesirably,  requires multiple exploratory RLE passes. 

Another popular workaround tackles this efficiency drawback via a heuristic: symbols with higher frequencies are more likely to be suitable to RLE. Therefore, it obtains the input distribution (with at most one pass) and only handles frequent symbol(s) with RLE~\cite{gandolph2010method,rle-preprocess-dcc2021,jang2023zero}. However, determining the \textit{frequency threshold} entirely depends on intuitions and human experience, which is not robust and often leads to size inflation. 
On one hand, if this threshold is too high and some common symbols are not encoded with RLE, the compression effect can be worse than the vanilla RLE. On the other hand, if RLE handles some relatively rare symbols, the compression performance may again get unideal, e.g., in the worst-case scenarios where all symbols have low frequency, even the most dominant symbol may not be suitable to RLE.

In this paper, we theoretically show that, for an arbitrary symbol, its suitability to RLE can be elegantly determined by joining its frequency and bit width (see Eq.~\eqref{eq:final}; for the sake of brevity, we assume in our derivation that the input symbols are Independently and Identically Distributed (i.i.d.)).
Based on this insight, we propose a novel method that \textit{automatically} calculates the aforementioned frequency threshold, and only the symbols that meet this threshold are encoded with RLE. To the best of our knowledge, no one before us has made such an exploration. Next, on both i.i.d. and non-i.i.d. testbeds that cover different real-world modalities (tabular data, time-series, and image), we empirically validate that our algorithm, which retains the efficiency advantage, substantially outperforms existing RLE baselines and consistently avoid the size inflation issue. 


\section{Preliminaries}\label{sec:math}

\paragraph{Problem formulation.} 
Let $N$ be the length of the input sequence and $N_x$ be the amount of elements whose symbol is $x$. To decide whether symbol $x$ is suitable to RLE, we need to predict whether these $N_x$ elements require no extra bits in the RLE output compared with those in the input. 

To begin with, let $N_x - \mathbb{R}_x$ denote symbol $x$'s number of occurrences in the encoded variable. As the amount of the corresponding elements in the run-control variable is also $N_x - \mathbb{R}_x$ (see \cref{sec:intro}), encoding $x$ with RLE saves
\begin{equation}\label{eq:original_condition}
    b_x N_x - (b_x + b_r) (N_x - \mathbb{R}_x)  
\end{equation}
bits in total, where $b_x$ stands for the number of bits used to represent the symbol $x$. It is clear that to decide whether $\text{Eq.~\eqref{eq:original_condition}} \geq 0$ holds, the key task is to compute $\mathbb{R}_x$.

\begin{figure}
    \centering
    \includegraphics[width=0.7\textwidth]{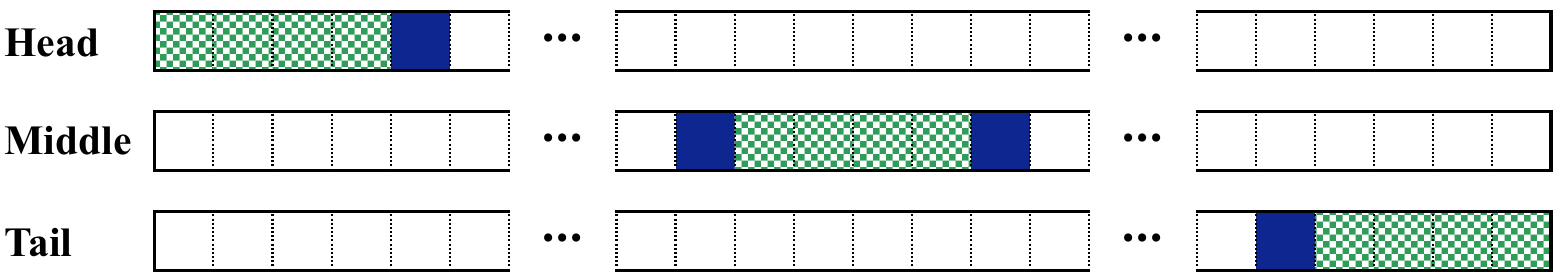}
    \caption{Three possible locations of a run where all symbols are $x$ and the length is precisely $n\leq N-2$ (in this example $n=4$). A \textcolor{ForestGreen}{green} (dotted hatch) cell stands for a symbol $x$, while \textcolor{Blue}{blue} (solid hatch) cells denote other symbols.}
    \label{fig:n-run-prop}
\end{figure}

\paragraph{Calculating $\mathbb{R}_x$.}
Following previous data compression works~\cite{iid-1995,iid-dcc2011,iid-2023}, we assume that the input sequence is generated from an i.i.d. source. When $N$ is sufficiently large, the probability that symbol $x$ occurs can be written as $p_x = \frac{N_x}{N}$. We can thus estimate the count of runs whose symbols are all $x$ and length is \textit{exactly} $n$:

\begin{itemize}[noitemsep, topsep=0pt, partopsep=0pt, label={$\bullet$}, leftmargin=10pt]
    \item When $n \leq N-2$, such a run may appear at three different types of locations in the input sequence (see Fig.~\ref{fig:n-run-prop}). In the ``middle'' case, such a run has $N-n-1$ possible starting indexes. For each index, its occurrence probability is $\textcolor{Blue}{(1-p_x)}\textcolor{ForestGreen}{p_x ^n}\textcolor{Blue}{(1-p_x)}$. Hence, the expected number of occurrence is $\textcolor{Blue}{(1-p_x)}\textcolor{ForestGreen}{p_x ^n}\textcolor{Blue}{(1-p_x)}(N-n-1)$. Similarly, the expected number of such a run's occurrence is $\textcolor{ForestGreen}{p_x ^n}\textcolor{Blue}{(1-p_x)}$ in the ``head'' case and $\textcolor{Blue}{(1-p_x)}\textcolor{ForestGreen}{p_x ^n}$ in the ``tail'' case.
    \item When $n=N-1$, such a run may only appear at the sequence's ``head'' or ``tail''.
    \item When $n=N$, trivially the number of occurrences is $p_x^N$. 
\end{itemize}  
It is known that encoding a run with a length of exactly $n$ requires $\lceil \frac{n}{2^{b_r}}\rceil$ run-control variable elements (see \cref{sec:intro}). Accordingly, we can show that
\begin{align}
    \mathbb{R}_{x}
    & = 
    \begingroup\color{lightgray}\underbrace{\color{black}\sum_{n=1}^{N-2} \big( \begingroup\color{lightgray}\overbrace{\color{black}(1 - p_x) p_x^n (1-p_x)(N-1-n)}^{\text{``middle''}}\endgroup + \begingroup\color{lightgray}\overbrace{\color{black}p_x^n (1-p_x)}^{\text{``head''}}\endgroup + \begingroup\color{lightgray}\overbrace{\color{black}(1-p_x)p_x^n}^{\text{``tail''}}\endgroup \big) \big( n-\lceil \frac{n}{2^{b_r}} \rceil \big)}_{n \leq N-2}\endgroup\nonumber \\
    & \hspace{5em} + \begingroup\color{lightgray}\underbrace{\color{black}\big( p_x^{N-1} (1-p_x) + (1-p_x)p_x^{N-1} \big) \big( N-1-\lceil \frac{N-1}{2^{b_r}} \rceil \big)}_{n=N-1 \text{ (only ``head'' and ``tail'')}}\endgroup + \begingroup\color{lightgray}\underbrace{\color{black}p_x^N \big( N-\lceil \frac{N}{2^{b_r}}\rceil \big)}_{n=N}\endgroup \nonumber \\
    & = p_x^N \big( N-\lceil \frac{N}{2^{b_r}}\rceil \big) + 
\sum_{n=1}^{N-1} \big( 2(1-p_x)p_x^n +  (1 - p_x)^2 p_x^n(N-1-n)\big) \big( n-\lceil \frac{n}{2^{b_r}} \rceil \big) \nonumber
\end{align}
For clarity, we rewrite this equation as
\begin{equation}\label{eq:real_sum}
    \mathbb{R}_x  = \sum_{n=1}^{N-1} (k_\alpha n + k_\beta)p_x ^n(n-1) + \epsilon_1 + \epsilon_2 
\end{equation}
where $k_\alpha  = -( 1-p_x )^2$, $k_\beta =\big( 1-p_x \big) \big( (1-p_x)(N-1)+2 \big)$, $\epsilon_1 =p_x^N\big( N-\lceil \frac{N}{2^{b_r}}\rceil \big)$, and $\epsilon_2  = \sum_{n=1}^{N-1} \big( k_\alpha n + k_\beta \big) p_x ^n \big( 1 - \lceil \frac{n}{2^{b_r}}\rceil \big)$.

\colorbox{gray!20}{If $p_x=1$, i.e., all input symbols are $x$}, when $N \gg 2^{b_r}$ (which holds in practice), via Eq.~\eqref{eq:real_sum} trivially we have
\begin{equation}
    \mathbb{R}_x = N - \lceil \frac{N}{2^{b_r}} \rceil = (1 - \frac{1}{2^{b_r}}) N \nonumber
\end{equation}
In this case, $x$ is suitable to RLE as long as
\begin{equation}\label{eq:condition_all}
    \text{Eq.~\eqref{eq:original_condition}} \geq 0 \iff b_x N \geq   \big( b_x + b_r \big) \big( N - (1 - \frac{1}{2^{b_r}}) N \big) \iff b_x \geq \frac{b_r}{2^{b_r} - 1}
\end{equation}
In practice, $b_r \geq 1$, so $1 \geq \frac{b_r}{2^{b_r} - 1}$; also $b_x \geq 1$, thus Eq.~\eqref{eq:condition_all} constantly holds. As a result, \ul{if all input elements are the same, RLE is always a suitable compressor}.

\colorbox{gray!20}{Otherwise, i.e., $p_x \in (0,1)$},
although we can still exploit Eq.~\eqref{eq:real_sum} directly to calculate $\mathbb{R}_x$, the computation overhead is too much. Concretely speaking, via Stirling's formula~\cite{Stirling} we know that the time complexity of $\sum_{n=1}^{N-1} (\cdot)^n$ is at least $\mathcal{O}(N \log N)$. 

\begin{figure}
    \centering
    \includegraphics[width=\textwidth]{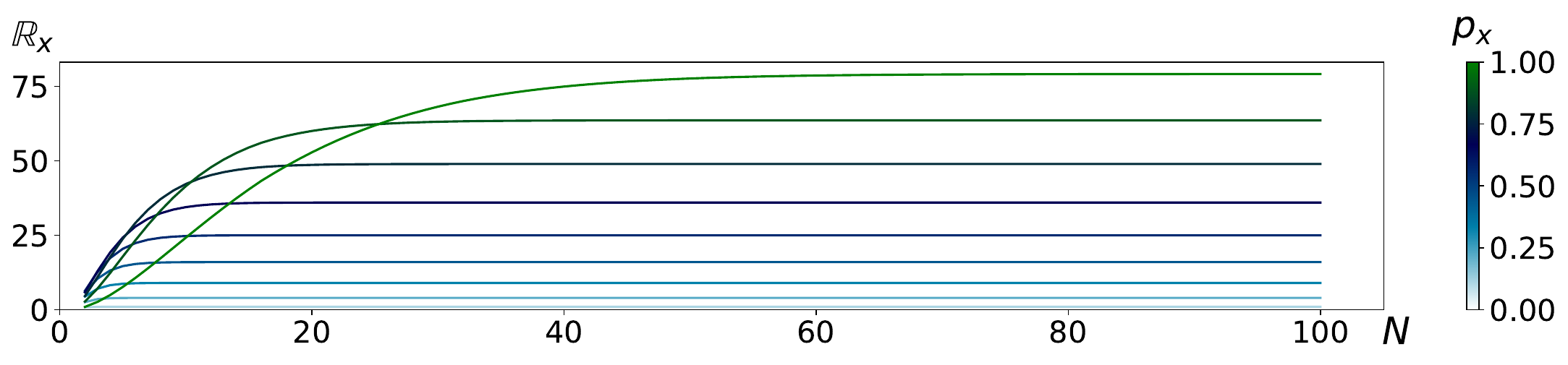}
    \caption{The relationship between $\mathbb{R}_x$ and $N$, with $b_r$ fixed at 4.}
    \label{fig:convergence}
\end{figure}
In Fig.~\ref{fig:convergence}, we notice that when $N$ surpasses 80, $\mathbb{R}_x$ displays a pronounced tendency towards convergence. In practice, $N$ is typically orders of magnitude larger than 80. Therefore, we can simplify the calculation of $\mathbb{R}_x$ by solving its limit. To be detailed, we first transform the summation of finite series to that of finite series as
\begin{align}
\mathbb{R}_x  & = \sum_{n=1}^{\infty} (k_\alpha n + k_\beta)p_x ^n(n-1) + \epsilon_1 + \epsilon_2 + \epsilon_3\nonumber \\
& = k_\alpha \sum_{n=1}^{\infty} p_x^n n^2 + (k_\beta - k_\alpha) \sum_{n=1}^{\infty} p_x^n  n - k_\beta \sum_{n=1}^{\infty} p_x^n + \epsilon_1 + \epsilon_2 + \epsilon_3 \nonumber
\end{align}
where $
\epsilon_3 = - \sum_{n=N}^{\infty} (k_\alpha n + k_\beta)p_x ^n(n-1)
$.
Then, using the three lemmas justified in the Appendix,
we can show that
\begin{align}
    \mathbb{R}_x  
& = k_\alpha \big( \frac{p_x^2 + p_x}{(1-p_x)^3} - 0\big)+ \big( k_\beta - k_\alpha \big) \big( \frac{p_x}{(1-p_x)^2} - 0 \big)- k_\beta \big( \frac{1}{1-p_x} - 1\big) + \epsilon_1 + \epsilon_2 + \epsilon_3 \nonumber\\
& =  p_x^2 (N-1) + \epsilon_1 + \epsilon_2 + \epsilon_3 \nonumber
\end{align}

\paragraph{Approximation analysis.} We show that $\epsilon_1$ can be neglected when $N$ is sufficiently large and $p_x \in (0, 1)$. Given $f(x) = \frac{\ln x}{x}$ is monotonically increasing, we have
\begin{equation}
        \frac{\ln(p_xN)}{p_xN} < \frac{\ln(N)}{N}  \iff N ln (p_x N) < (p_x N)lnN  \iff (p_x N)^{N} < N^{p_x N} \nonumber 
\end{equation}
Because $N - p_x N  = N - N_x \geq 1$,
\begin{equation}
    \epsilon_1 = p_x^N\big( N-\lceil \frac{N}{2^{b_r}} \rceil \big) < p_x^N N = \frac{(p_x N)^N}{N^{p_x N}} \cdot \frac{N}{N^{N- p_x N}} < \frac{N}{N^{N- p_x N}} \leq1 \ll N \nonumber
\end{equation}

\begin{figure}
    \centering
    \includegraphics[width=\textwidth]{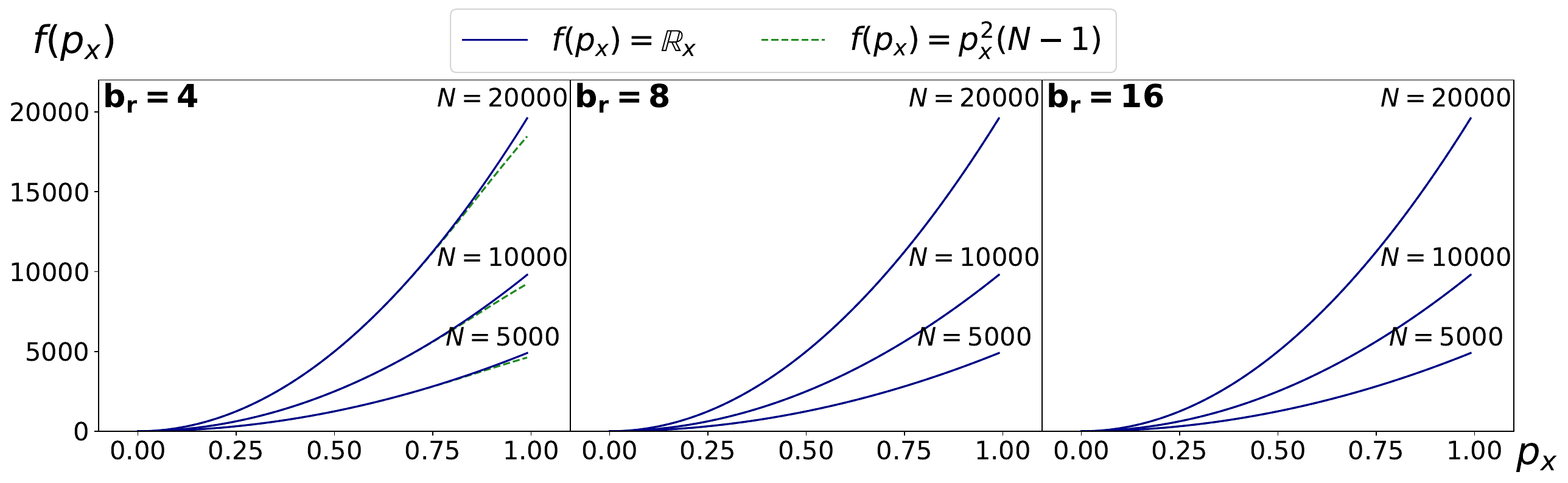}
    \caption{Comparisons between $\mathbb{R}_x$ and $p_x^2(N-1)$ under different configurations.}
    \label{fig:func_compare}
\end{figure}
As for $\epsilon_2$ and $\epsilon_3$, combining them yields
\begin{equation}\label{eq:e2+e3}
    \epsilon_2 + \epsilon_3 = \sum_{n=1}^{\infty} (k_\alpha n + k_\beta)p_x ^n - \big[ \sum_{n=1}^{N-1} (k_\alpha n + k_\beta)p_x ^n \lceil \frac{N}{2^{b_r}} \rceil + \sum_{n=N}^{\infty} (k_\alpha n + k_\beta)p_x ^n n \big] 
\end{equation}
Instead of quantifying the complex Eq.~\eqref{eq:e2+e3}, we investigate its impact by directly plotting $\mathbb{R}_x$ and $p_x^2(N-1)$ for comparison. 
As shown in Fig.~\ref{fig:func_compare}, in settings with various $b_r$ and $N$, $\mathbb{R}_x$ and $p_x^2(N-1)$ almost coincide perfectly when $b_r\geq 8$. In the case that $b_r = 4$, their trends still match well, especially when $p_x \leq 0.75$. As for $p_x > 0.75$, i.e., most input symbols are $x$, it is almost certain that RLE serves as a suitable compressor for $x$ and solving Eq.~\eqref{eq:original_condition} is no longer necessary. We thus argue that the minor difference between $\mathbb{R}_x$ and $p_x^2(N-1)$ is negligible in practice. 
Moreover, based on the observation that $\mathbb{R}_x \gtrsim p_x^2(N-1)$, we have
\begin{equation}
    p_x^2 (b_x + b_r) (N-1) \geq N_x \cdot b_r \implies (b_x + b_r) \mathbb{R}_x \geq b_r N_x \implies 
    \text{Eq.~\eqref{eq:original_condition}} \geq 0 \nonumber
\end{equation}
i.e., in terms of guaranteeing that encoding $x$ with RLE will not lead to more considerable storage usage, substituting $\mathbb{R}_x$ with $p_x^2(N-1)$ in Eq.~\eqref{eq:original_condition} will never lead to incorrect decisions. What is even better, this approximation can significantly reduce the time complexity of the computing $\mathbb{R}_x$ from $\mathcal{O}(N \log N)$ to $\mathcal{O}(1)$. 

Finally, when $p_x \in (0,1)$, the threshold of identifying $x$ as a symbol that RLE may encode can be derived as
\begin{align}
    p_x^2 (b_x + b_r)(N-1)  \geq b_r N_x  & \iff (b_x + b_r) (N-1)(\frac{N_x}{N})^2 \geq b_r N_x \nonumber \\ &\iff N_x \geq \frac{b_r}{b_x + b_r} \cdot \frac{N^2}{N-1} \nonumber
\end{align}
which, given $N$ is sufficiently large in practice, can be simplified as
\begin{equation}\label{eq:final}
    N_x \geq \frac{b_r}{b_x + b_r} \cdot N \iff p_x \geq \frac{b_r}{b_x + b_r} 
\end{equation}
Remarkably, since Eq.~\eqref{eq:final} always holds when $N_x=N$, it also works when $p_x=1$.

\section{Algorithm}
\paragraph{Compression.} In \cref{sec:math}, we identified the criterion to assess the suitability of applying RLE on specific symbols. Based on this, we propose compressing only a subset of the input sequence while leaving other symbols unaffected. To be exact,
\begin{itemize}[noitemsep, topsep=0pt, partopsep=0pt, label={$\bullet$}, leftmargin=10pt]
    \item \textbf{Step 1: Constructing set $\mathcal{G}$.} We add $x$ to $\mathcal{G}$ (a set) if Eq.~\eqref{eq:final} holds, so that $\mathcal{G}$ contains all symbols suitable to RLE. In case the symbols' frequencies are unknown, we randomly sample elements from the input to infer the distribution. 
    \item \textbf{Step 2: Encoding.} Scanning the entire input sequence, if the occurred symbol can be found in $\mathcal{G}$, our algorithm will store it in the encoded variable and the count of its consecutive repeats in the run-control variable, i.e., identical to how the vanilla RLE behaves. Otherwise, our approach will directly append this symbol to the encoded variable, regardless of its successor. For instance, with $\mathcal{G}=\{\texttt{\textcolor{ForestGreen}{a}}, \texttt{\textcolor{Blue}{b}}\}$, \texttt{<\textcolor{ForestGreen}{a}, \textcolor{Blue}{b}, \textcolor{Blue}{b}, \textcolor{Blue}{b}, \textcolor{ForestGreen}{a}, \textcolor{ForestGreen}{a}, \textcolor{red}{c}, \textcolor{red}{c}, ...>} (the same example sequence in \cref{sec:intro}) will be rendered as \texttt{<\textcolor{ForestGreen}{a}, \textcolor{Blue}{b}, \textcolor{ForestGreen}{a}, \textcolor{red}{c}, \textcolor{red}{c}, ...>} and \texttt{<\textcolor{ForestGreen}{1}, \textcolor{Blue}{3}, \textcolor{ForestGreen}{2}, ...>}. 
\end{itemize}

\vspace{-9pt}
\paragraph{Decompression.} 
Our algorithm scans the encoded variable. If a symbol belonging to $\mathcal{G}$ occurs, it will be repeated for $r$ times in the reconstructed sequence, where $r$ denotes the first \textit{unvisited} element in the run-control variable. Otherwise, this symbol will be directly appended to the reconstructed sequence without repetition.

\vspace{-9pt}
\paragraph{Discussion.} As $\mathcal{G}$ only contains sufficiently frequent symbols, its cardinality (which is normally below a dozen in practice), as well as size, can be neglected. Therefore, as long as the input sequence satisfies the i.i.d. assumption, our scheme can lead to no worse compression than the vanilla RLE. More importantly, in theory, the output of our algorithm is guaranteed to be no larger in size than the input.

In practice, increasing the sample size at the compression stage may lead to a better estimation on the distribution of the entire input sequence, while consuming more computational resources. When the sequence is long and segmentation is needed before compression, sampling should be done on each segment. In any case, obtaining the input distribution needs at most one full pass, so the time complexity of constructing $\mathcal{G}$ is $\mathcal{O}(N)$. Because the encoding process also takes $\mathcal{O}(N)$, the overall compression time complexity is still $\mathcal{O}(N)$. Similarly, the decompression time complexity is $\mathcal{O}(N)$. To sum up, our algorithm substantially enhances the compression performance of the vanilla RLE, with the efficiency advantage retained. 


\section{Experiments}

\begin{figure}
    \centering
    \includegraphics[width=\textwidth]{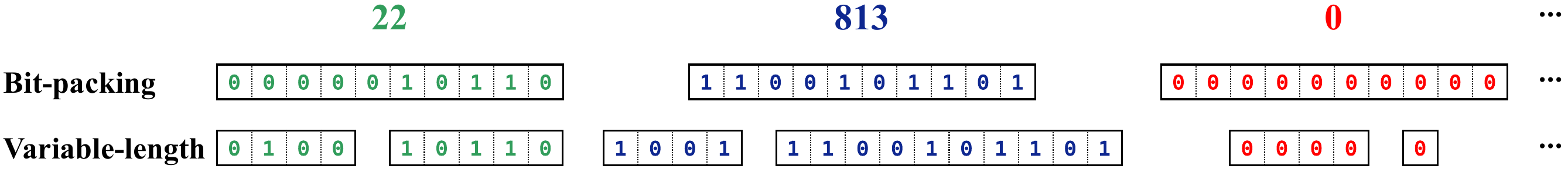}
    \caption{Examples of the two binary representation schemes we consider in \cref{ssec:exp-setup} (all elements are less than $2^{10}$). NB: for each variable-length code, the components on the left and right are $c_\alpha$ and $c_\beta$, respectively.}
    \label{fig:binary}
\end{figure}

\subsection{Setups}\label{ssec:exp-setup}
\paragraph{Implementation details.} The lower $b_r$ is, the more likely Eq.~\eqref{eq:final} holds, i.e., the stronger RLE's compression power on $x$ tends to be. Therefore, by default, we set $b_r$ at 4, the minimum configuration widely used in literature~\cite{images,daily2010data,KLEIN2014326}. When dealing with each input sequence, we randomly sample 10K elements to generate $\mathcal{G}$. As for the binary representation of symbols, we consider two widely adopted strategies, as illustrated in Fig.~\ref{fig:binary}:
\begin{itemize}[noitemsep, topsep=0pt, partopsep=0pt, label={$\bullet$}, leftmargin=10pt]
    \item \textbf{Bit-packing}~\cite{bit-packing}: Each symbol is represented by $b$ bits, with $b$ being the minimum number of bits required to encode any symbol.
    \item \textbf{Variable-length}~\cite{variable-length}: Each symbol is encoded into two components, namely $c_\alpha$ and $c_\beta$. $c_\alpha$ stores the width of $c_\beta$ (empirically we fix the width of $c_\alpha$ at 4 bits). $c_\beta$ represents the symbol's actual value, using the fewest number of bits necessary. 
\end{itemize}

\vspace{-9pt}
\paragraph{Testbeds.}
We utilise the following resources, which cover three distinct modalities:
\begin{itemize}[noitemsep, topsep=0pt, partopsep=0pt, label={$\bullet$}, leftmargin=10pt]
    \item \textbf{SO-2018\footnote{\url{https://insights.stackoverflow.com/survey/2018}}} stores the responses in a survey conducted in 2018 by Stack Overflow, a famous online developer community. This tabular archive contains answers to 128 queries (such as personal background, skill set, and job preferences) in the questionnaire, which are processed as 128 input sequences in our experiment. Since the 0.1 million rows (each corresponds to one respondent) have been randomly permuted for anonymisation purposes, this dataset can be regarded as a prime example of (approximately) i.i.d. sequences in the real world. 
    \item \textbf{Netstream} is a private dataset held by Huawei. It comprises traffic monitoring logs collected via our internal telecommunication devices. We focus on 5 landmark entries, each of which includes over 17 million records. Please be aware that this large-scale time-series dataset does not satisfy the i.i.d. assumption.
    \item \textbf{``Lena''\footnote{\url{http://www.lenna.org/lena_std.tif}}} is the premier standard test image in the field of data compression. To flatten it, for simplicity, we adopt a workflow motivated by the compression procedure of the BMP format~\cite{images}. To be concrete, we first convert the original image to a bitmap and then concatenate all pixel lines. As a result, ``Lena'' (originally boasts a resolution of 512$\times$512) is transformed into a sequence with 262,144 elements.
\end{itemize}

Remarkably, many symbols in SO-2018 and Netstream are long strings. Following
~\cite{Sayood1996IntroductionTD,gandolph2010method,iid-2023}, within every sequence, we replace each symbol with an integer identifier based on its order of appearance (starting from 0 and incremented cumulatively).

\vspace{-9pt}
\paragraph{Studied methods.} Besides the proposed algorithm (abbreviated as \textbf{\textsf{Ours}}), we additionally benchmark two baselines: \textbf{\textsf{V-RLE}} (aka the vanilla RLE), which handles all symbols with RLE; \mbox{\textbf{\textsf{D-RLE}}}, which only encodes the most
dominant symbol with RLE (mentioned in \cref{sec:intro}; it can be regarded as a special case of \textsf{Ours}, where $\mathcal{G}$ always has the most
dominant symbol as its only item).

\begin{figure}
     \centering
     \includegraphics[width=0.495\textwidth]{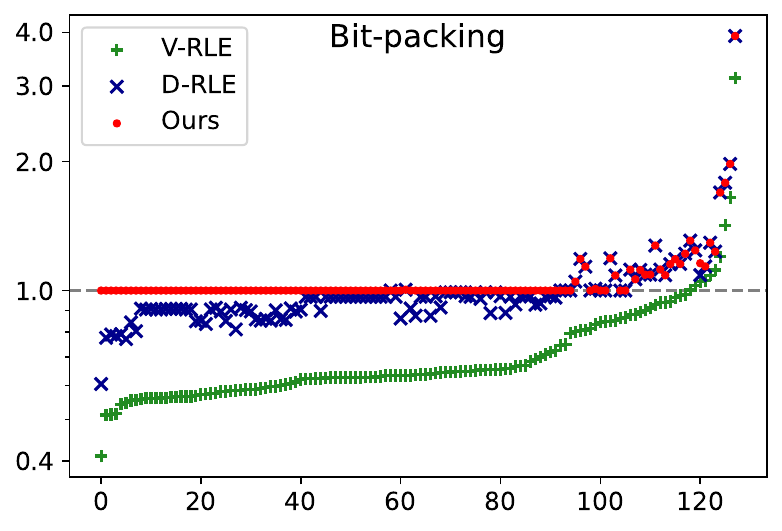}
\includegraphics[width=0.495\textwidth]{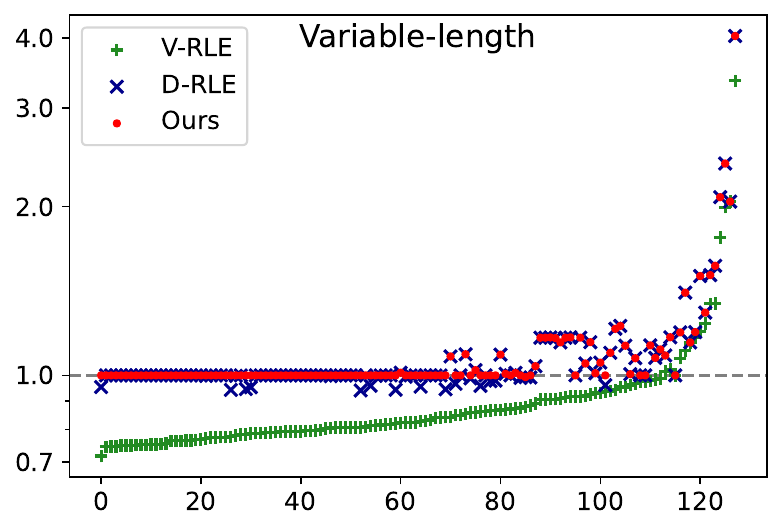}
     \caption{Column-wise log-scaled Compression Ratio (CR, the higher the better) on the SO-2018 dataset. Data points with the same X coordinate correspond to the same column (for visibility, the 128 columns are sorted according to the CR of \textsf{V-RLE}).}
     \label{fig:so-2018}
\end{figure}
\setlength{\tabcolsep}{0.1cm}
\begin{table}[!t]\footnotesize
\centering
\begin{tabular}{p{1.4cm}?>
{\lcell}p{0.8cm}>{\rcell}p{1.05cm}|>{\lcell}p{0.8cm}>{\rcell}p{1.05cm}|>{\lcell}p{0.8cm}>{\rcell}p{1.05cm}?>{\lcell}p{0.8cm}>{\rcell}p{1.05cm}|>{\lcell}p{0.8cm}>{\rcell}p{1.05cm}|>{\lcell}p{0.8cm}>{\rcell}p{1.05cm}}
\toprule 
  & \multicolumn{6}{c?}{Bit-packing} & \multicolumn{6}{c}{Variable-length}\\ \cmidrule{2-13}
 & \multicolumn{2}{c|}{\textsf{V-RLE}}&  \multicolumn{2}{c|}{\textsf{D-RLE}} & \multicolumn{2}{c?}{\textsf{Ours}} & \multicolumn{2}{c|}{\textsf{V-RLE}} & \multicolumn{2}{c|}{\textsf{D-RLE}} & \multicolumn{2}{c}{\textsf{Ours}} \\ \toprule
 \texttt{Src\_IP} & 409.5 &[-68.3] & 364.6 & [-23.4] & \bftab 341.2 &[0.0] & 426.5 & [-41.7] & 423.8 &[-39.0] & \bftab 384.8 &[0.0]\\ \cmidrule{1-13}
\texttt{Dst\_IP} & 334.9 &[-10.1] & 308.4 &[16.4] & \bftab 303.5 & [21.3] & 428.5 &[-76.5] & 421.6 & [-69.6] & \bftab 352.0 & [0.0]\\ \cmidrule{1-13}
\texttt{Src\_Port} & 282.3 &[-9.3] & 280.6 &[-7.6] & \bftab 273.0 & [0.0] & 315.4 & [-25.9] & 304.2 &[-14.7] & \bftab 289.5 &[0.0]\\ \cmidrule{1-13}
\texttt{Dst\_Port} & 227.0 &[-36.0] & 205.8 &[-14.8] & \bftab 191.0 &[0.0] & 216.2 &[14.6] & 213.6 &[17.2] & \bftab 198.0 &[32.8]\\ \cmidrule{1-13}
\texttt{Protocol} & \textcolor{white}{0}66.2 &[10.1] & \textcolor{white}{0}65.5 &[10.8] & \bftab  \textcolor{white}{0}63.2 &[13.1] & \textcolor{white}{0}68.0 &[3.3] & \textcolor{white}{0}61.6 &[6.4] &  \bftab\textcolor{white}{0}60.7 &[7.3]\\
                           \bottomrule
\end{tabular}
\caption{Entry-wise size (MB) of Netstream after encoding. In both this table and Tab.~\ref{tab:lena}, we highlight the smallest size of each setting in {\bftab bold}, and additionally calculate the \textit{reduced} size (see the square bracket on the right of every cell).}\label{tab:netstream}
\end{table}

\subsection{Results}

\paragraph{SO-2018.} To provide clear and unified results on all the 128 columns and both binary representation settings, as demonstrated in Fig.~\ref{fig:so-2018}, we report the Compression Ratio (CR), i.e., $\frac{\text{input size}}{\text{output size}}$ (the higher, the better). Overall, we observe that despite achieving successful compression sometimes (see data points on the right of each sub-figure), \textsf{V-RLE} struggles on most columns. This is because most columns of SO-2018 have high cardinality, where even the most dominant value has a low frequency. Implied from \cref{sec:math}, under such a distribution, the number of occurrences of long runs is low, and \textsf{V-RLE} is thus unsuitable. As for the existing workaround, \textsf{D-RLE}, it does address this challenge to some degree: on most columns, it yields higher CR than \textsf{V-RLE}. Nevertheless, we still see many columns where the output of \textsf{D-RLE} is undesirably larger than the input in size. In contrast, the CR of \textsf{Ours} is not only on par with or above that of either counterpart on \textit{all} columns, but also \textit{always} equal to or higher than 1.0. In other words, our experiments confirm that \textsf{Ours}, which substantially outperforms \textsf{V-RLE} and \textsf{D-RLE}, consistently avoids the size inflation issue when the i.i.d. assumption is (approximately) satisfied.

\vspace{-14pt}
\paragraph{Netstream.} Although in \cref{sec:math} we assume the input to be i.i.d., on non-i.i.d. data such as Netstream, we still find our method effective.  As displayed in Tab.~\ref{tab:netstream}, \textsf{Ours} consistently compresses the time-series to the smallest size. Also, it is the only approach that never requires more bits for the output than the input. In comparison, on all the 10 examined settings, \textsf{V-RLE} and \textsf{D-RLE} end up with negative storage space reduction for 7 and 6 times, respectively. 


\newcommand{\shiftleft}[2]{\makebox[0pt][r]{\makebox[#1][l]{#2}}}

\setlength{\tabcolsep}{0.1cm}
\begin{table}[!t]\footnotesize
\centering
\begin{tabular}{p{2.5cm}?>{\ccell}p{1.9cm}|>{\ccell}p{2.7cm}|>{\ccell}p{0.5cm}?>{\lcell}p{0.8cm}>{\rcell}p{1.2cm}|>{\lcell}p{0.8cm}>{\rcell}p{1.cm}|>{\lcell}p{0.8cm}>{\rcell}p{1.cm}}
\toprule 
 \multirow{6}{*}{\shiftleft{0.1cm}{
 \raisebox{-0.9\totalheight}{\includegraphics[width=2.5cm]{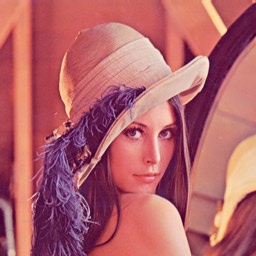}}
 }}
& \# of colours& $b_x$ (colour depth) & $b_r$ &  \multicolumn{2}{c|}{\textsf{V-RLE}}&  \multicolumn{2}{c|}{\textsf{D-RLE}} & \multicolumn{2}{c}{\textsf{Ours}} \\ \cmidrule[1pt]{2-10}
& 16 & 4 & 4 & 238.1 & [-110.1] & 135.2 & [-7.2] & \bftab 128.0 & [0.0] \\ \cmidrule{2-10}
& 16 & 8 & 4 & 363.4 & [-107.4] & 257.3 & [-1.3] & \bftab 256.0 & [0.0] \\ \cmidrule{2-10}
& 256 & 8 & 8 & 484.5 & [-228.5] & 258.8 & [-2.8] & \bftab 256.0 & [0.0] \\ \cmidrule{2-10}
& 65,536 & 16 & 4 & 635.8 & [-123.8] & 512.1 & [-0.1] & \bftab 512.0 & [0.0] \\ \cmidrule{2-10}
& 65,536 & 16 & 8 & 762.9 & [-250.9] & 512.2 & [-0.2] & \bftab 512.0 & [0.0] \\ 
\bottomrule
\end{tabular}
\caption{Size (KB) of ``Lena'' after encoding.}\label{tab:lena}
\end{table}

\vspace{-14pt}
\paragraph{``Lena''.} We exploit three popular bitmap palette configurations, which respectively require 4, 8, and 16 bits (aka \textit{colour depth}) for each pixel, so as to support 16, 256, and 65,536 colours. Following the BMP standard~\cite{images}, we fix the binary representation at bit-packing, and consider $b_r \in \{4, 8\}$ with $b_r \leq b_x$ for RLE. Results in Tab.~\ref{tab:lena} re-verify the superiority of \textsf{Ours} when compressing non-i.i.d. input. To be concrete, \textsf{V-RLE} increases the storage space in all setups by large margins. While this problem tends to be less severe on \textsf{D-RLE}, we still observe size inflation, especially on smaller $b_x$. \textsf{Ours}, on the contrary, identifies no suitable symbol for $\mathcal{G}$, so that RLE is not actually applied. Therefore, although \textsf{Ours} yields zero compression effect on ``Lena'', it guarantees that the size does not explode, which is critical in practice.

\section{Conclusion and Future Work}
This paper concerns the long-standing size inflation challenge of RLE. We first derive an elegant equation to automatically identify symbols suitable to RLE, based on which we develop an algorithm that selectively handles part of the input sequence with RLE. Next, on real-world tabular, time-series, and image testbeds, we empirically validate that (1) the proposed method constantly avoids size inflation, and (2) it is substantially superior to existing RLE techniques.

While our experimental results on non-i.i.d. datasets are promising, the theoretical analysis presented in \cref{sec:math} is restricted to the i.i.d. context. In the future, we will extend our theory to more complex distributions, as well as testify the proposed algorithm on more diverse applications.

\Section{References}
\bibliographystyle{IEEEbib}
\bibliography{main}

\begin{thebibliography}{10}

\bibitem{Golomb1966RunlengthE}
Solomon~W. Golomb,
\newblock ``Run-length encodings ({C}orresp.),''
\newblock {\em IEEE Transactions on Information Theory}, vol. 12, pp. 399--401, 1966.

\bibitem{KLEIN2014326}
Shmuel~T. Klein and Dana Shapira,
\newblock ``Practical fixed length {L}empel-{Z}iv coding,''
\newblock {\em Discrete Applied Mathematics}, vol. 163, pp. 326--333, 2014,
\newblock Stringology Algorithms.

\bibitem{jang2023zero}
Myeongjae Jang, Jinkwon Kim, Haejin Nam, and Soontae Kim,
\newblock ``Zero and narrow-width value-aware compression for quantized convolutional neural networks,''
\newblock {\em IEEE Transactions on Computers}, 2023.

\bibitem{rle-preprocess-dcc2021}
Sven Fiergolla and Petra Wolf,
\newblock ``Improving run length encoding by preprocessing,''
\newblock in {\em 2021 Data Compression Conference (DCC)}. IEEE, 2021, pp. 341--341.

\bibitem{sprintz}
Davis~W. Blalock, Samuel Madden, and John~V. Guttag,
\newblock ``Sprintz: Time series compression for the internet of things,''
\newblock in {\em ACM on Interactive, Mobile, Wearable and Ubiquitous Technologies (IMWUT)}, 2018.

\bibitem{images}
John Miano,
\newblock {\em Compressed image file formats - JPEG, PNG, GIF, XBM, {BMP}},
\newblock Addison-Wesley-Longman, 1999.

\bibitem{Sayood1996IntroductionTD}
Khalid Sayood,
\newblock {\em Introduction to data compression},
\newblock Morgan Kaufmann, 1996.

\bibitem{gandolph2010method}
Dirk Gandolph, Jobst H{\"o}rentrup, Axel Kochale, Ralf Ostermann, and Hartmut Peters,
\newblock ``Method for run-length encoding of a bitmap data stream,'' 2010,
\newblock US Patent 7,657,109.

\bibitem{iid-1995}
Kar-Ming Cheung and Aaron~B. Kiely,
\newblock ``An efficient variable length coding scheme for an {IID} source,''
\newblock in {\em Data Compression Conference ({DCC})}, 1995.

\bibitem{iid-dcc2011}
Mark~Z. Mao, Robert~M. Gray, and Tam{\'{a}}s Linder,
\newblock ``On asymptotically optimal stationary source codes for {IID} sources,''
\newblock in {\em Data Compression Conference {(DCC})}, 2011.

\bibitem{iid-2023}
Daniel Severo, James Townsend, Ashish Khisti, Alireza Makhzani, and Karen Ullrich,
\newblock ``Compressing multisets with large alphabets,''
\newblock {\em IEEE Journal on Selected Areas in Information Theory}, 2023.

\bibitem{Stirling}
Philippe Flajolet and Robert Sedgewick,
\newblock {\em Analytic Combinatorics},
\newblock Cambridge University Press, 2009.

\bibitem{daily2010data}
Kenny Daily, Paul Rigor, Scott Christley, Xiaohui Xie, and Pierre Baldi,
\newblock ``Data structures and compression algorithms for high-throughput sequencing technologies,''
\newblock {\em BMC Bioinformatics}, vol. 11, pp. 1--12, 2010.

\bibitem{bit-packing}
Hao Jiang, Chunwei Liu, John Paparrizos, Andrew~A Chien, Jihong Ma, and Aaron~J Elmore,
\newblock ``Good to the last bit: Data-driven encoding with {C}odec{DB},''
\newblock in {\em International Conference on Management of Data ({SIGMOD})}, 2021.

\bibitem{variable-length}
Latha Pillai,
\newblock ``Variable length coding,''
\newblock {\em Application Note: Virtex-II Series}, 2003.

\end{thebibliography}

\appendix

\Section{Appendix} 
\begin{lemma}\label{lemma-1}\normalfont
\begin{equation}
    \sum_{n=0}^\infty a^n= \frac{1}{1-a} \hspace{2em} \text{subject to} \hspace{2em} a \in (0, 1) \nonumber
\end{equation}
\end{lemma}

\begin{proof}
By summing up a Geometric Progression and finding its limit, for $a \in (0, 1)$, we can show that
\begin{align}
    &\sum_{n=0}^N a^n= \frac{a^0 (1 - a^N)}{1-a} \implies \sum_{n=0}^\infty a^n = \lim_{N \to \infty} \frac{a^0 (1 - a^N)}{1-a} = \frac{1}{1-a} 
    \qedhere
\end{align}
\end{proof}

\begin{lemma}\label{lemma-2}\normalfont
\begin{equation}
    \sum_{n=0}^\infty a^n n= \frac{a}{(1-a)^2} \hspace{2em} \text{subject to} \hspace{2em} a \in (0, 1) \nonumber
\end{equation}
\end{lemma}
\begin{proof}
    It can be justified by respectively calculating the derivatives of both sides of the equation in Lemma~\ref{lemma-1}, as
    \begin{align}
        & \frac{\mathrm{d}}{\mathrm{d} a} \sum_{n=0}^\infty a^n= \frac{\mathrm{d}}{\mathrm{d} a} \frac{1}{1-a}  \iff
    \sum_{n=0}^\infty a^{n-1} n = \frac{1}{(1-a)^2} \iff \sum_{n=1}^\infty a^{n} n = \frac{a}{(1-a)^2}  \qedhere 
    \end{align}
\end{proof}

\begin{lemma}\label{lemma-3}\normalfont
\begin{equation}
    \sum_{n=0}^\infty a^n n^2= \frac{a^2 + a}{(1-a)^3} \hspace{2em} \text{subject to} \hspace{2em} a \in (0, 1) \nonumber
\end{equation}
\end{lemma}

\begin{proof}
Similar to the previous proof, we simply differentiate both sides of the equation in Lemma~\ref{lemma-2},  yielding
\begin{align}\label{eq:lemma-2-middle}
    \frac{\mathrm{d}}{\mathrm{d} a} \sum_{n=0}^\infty a^n n= \frac{\mathrm{d}}{\mathrm{d} a} \frac{a}{(1-a)^2} & \iff \sum_{n=0}^\infty a^{n-1} n^2 = \frac{(1-a)^2+2a(1-a)}{(1-a)^4}= \frac{a+1}{(1-a)^3} \nonumber \\
    & \iff \sum_{n=0}^\infty a^n n^2= \frac{a^2 + a}{(1-a)^3} \qedhere
\end{align}
\end{proof}

\end{document}